%% file: main.tex
\documentclass[5p]{elsarticle}
\usepackage{amsthm} 

\usepackage{amsmath}
\usepackage{amssymb}
\usepackage{algorithm}
\usepackage{tikz}
\usepackage[noend]{algpseudocode}
\usepackage{graphicx}
\usepackage{pgfplots}
\usepackage{mathtools}
\usepackage{subcaption}
\usepackage{marvosym}
\usepackage{fontawesome5}
\usepackage{centernot}
\usepackage{mathrsfs}
\usepackage{url}

\usepgfplotslibrary{fillbetween} 
\usepackage[colorinlistoftodos,prependcaption]{todonotes}

\algrenewcommand\textproc{\textsc}

\title{Efficiently Computing the Cyclic Output-to-Output Gain}

\author[aff1]{Daniel Arnström}
\author[aff1]{André M.H. Teixeira}
\address[aff1]{Department of Information Technology, Uppsala University, Sweden}

\begin{document}
\renewcommand{\baselinestretch}{1.0}

\definecolor{set19c1}{HTML}{E41A1C}
\definecolor{set19c2}{HTML}{377EB8}
\definecolor{set19c3}{HTML}{4DAF4A}
\definecolor{set19c4}{HTML}{984EA3}
\definecolor{set19c5}{HTML}{FF7F00}
\definecolor{set19c6}{HTML}{FFFF33}
\definecolor{set19c7}{HTML}{A65628}
\definecolor{set19c8}{HTML}{F781BF}
\definecolor{set19c9}{HTML}{999999}

\pgfplotstableread{data/hamiltonian_new.dat}{\hamiltonian}
\pgfplotstableread{data/mosek_new.dat}{\mosek}
\pgfplotstableread{data/nwc_hamiltonian.dat}{\ncshamiltonian}
\pgfplotstableread{data/nwc_mosek.dat}{\ncsmosek}
\pgfplotstableread{data/nwc_clarabel.dat}{\ncsclarabel}
\pgfplotstableread{data/reg_ex7.dat}{\exreg}
\pgfplotstableread{data/reg_unbounded.dat}{\exregunb}

\begin{abstract}
    The cyclic output-to-output gain is a security metric for control systems. Commonly, it is computed by solving a semi-definite program, which scales badly and inhibits its use for large-scale systems. We propose a method for computing the cyclic output-to-output gain using Hamiltonian matrices, similar to existing methods for the $H_\infty$-norm. In contrast to existing methods for the $H_{\infty}$-norm, the proposed method considers generalized singular values rather than regular singular values. Moreover, to ensure that the Hamiltonian matrices exist, we introduce a regularized version of the cyclic output-to-output gain. Through numerical experiments, we show that the proposed method is more efficient, scalable, and reliable than semi-definite programming approaches.
\end{abstract}


\maketitle
\thispagestyle{empty}
\pagestyle{empty}
\newtheorem{proposition}{Proposition}
\newtheorem{lemma}{Lemma}
\newtheorem{corollary}{Corollary}
\newtheorem{remark}{Remark}
\newtheorem{theorem}{Theorem}
\newtheorem{definition}{Definition}
\newtheorem{assumption}{Assumption}
\newtheorem{example}{Example}
\newtheorem{problem}{Problem}
\newtheorem{property}{Property}

\input{content.tex}

\bibliographystyle{elsarticle-num}
\bibliography{lib.bib}

\end{document}

%% file: content.tex
\section{Introduction}
As cyber-physical systems become more common, cybersecurity for control systems is becoming more important \cite[\S 4.C]{roadmap}. When considering possible cyber attacks, two important aspects are (i) the potential damage of an attack and (ii) if the attack can be detected \cite{teixeira2015secure}. 
A metric that combines the potential damage of an attack and its stealthiness is the  
\emph{output-to-output gain} (OOG) \cite{teixeira2021security}. 
When the OOG is used in practice, it is often computed by solving a semi-definite program (SDP) \cite{teixeira2021security,anand2023risk,nguyen2024scalable}. While manageable for small- to medium-scale control systems, this approach becomes restrictive for large-scale systems. Recent works have attempted to make the OOG more scalable \cite{nguyen2024scalable,anand2024scalable}, but are restricted in their applicability. In \cite{nguyen2024scalable}, the OOG for \textit{positive} systems are computed by restricting SDP decision variables to diagonal matrices \cite{nguyen2024scalable}, which improves scalability at the cost of an inexact metric. In \cite{anand2024scalable} the OOG is made more scalable by solving a linear program instead of an SDP; however, it also only applies to positive systems. For large-scale systems that are not positive, efficiently computing the OOG is an open problem. 

The OOG is a generalization of the classical $H_{\infty}$-norm, which plays a central role in robust control \cite{francis1987course}. 
For the $H_\infty$-norm, there exists efficient ways of computing it that are based on computing eigenvalues of Hamiltonian matrices \cite{boyd1989bisection,bruinsma1990fast}, rather than solving semidefinite programs. The same techniques have, however, not yet been applied to compute the OOG. What restsricts the methods in \cite{boyd1989bisection,bruinsma1990fast} to be directly applied to compute the OOG is twofold. First, a particular Algebraic Riccati Equation (ARE) associated with the OOG might not exist \cite{willems1971least}. Secondly, the OOG is more general than the $H_{\infty}$-norm, and hence the particular Hamiltonian matrix in \cite{boyd1989bisection,bruinsma1990fast} does not generally apply for the OOG.

The main contribution of this paper is a method for computing a regularized version of the output-to-output gain by computing eigenvalues of Hamiltonian matrices. The regularization ensures that a corresponding ARE to the regularized output-to-output gain always exists, which allows the Hamiltonian matrices in \cite{boyd1989bisection,bruinsma1990fast} to be extended to a more general setting.  
We show that the resulting method is several orders of magnitude faster than solving a corresponding semidefinite program. Numerical experiments for networked control systems also show that the proposed method is more scalable and reliable than the scalable method proposed in \cite{nguyen2024scalable}. In addition, the proposed method can, in contrast to \cite{nguyen2024scalable,anand2024scalable}, be applied to systems that are not positive.

In summary, the contributions of the paper are: 
\begin{enumerate}
    \item We introduce a \emph{regularized} version of the \emph{cyclic} OOG (Definition~\ref{def:roog}).
    \item We show how the regularized cyclic OOG is related to a Hamiltonian matrix (Theorem~\ref{th:hamiltonian} and Corollary~\ref{cor:ham-rcoog}).  
    \item We propose a method that efficiently computes the regularized cyclic OOG by computing eigenvalues to a Hamiltonian matrix (Algorithm \ref{alg:main}).
\end{enumerate}

The rest of the paper is structured as follows: in Section~\ref{sec:oog}, we introduce different version of the output-to-output gain, leading up to the regularized cyclic output-to-output gain. We then propose a method to efficiently compute the regularized cyclic output-to-output gain in Section~\ref{sec:comp}. The favourable properties of this method is then illustrated in Section~\ref{sec:result} through numerical experiments on general system and positive networked control systems. We then conclude that paper in Section \ref{sec:conc} with some possible avenues for future work.

\section{Output-to-output gains}
\label{sec:oog}
The main problem considered in this paper is how to compute so-called \emph{output-to-output gains} efficiently. In this section, we define different kinds of output-to-output gains and introduce corresponding semi-definite programs (SDPs) that are commonly used to compute these gains. This sets the stage for the main contribution in Section~\ref{sec:comp}, where we propose a novel way to compute output-to-output gains by computing eigenvalues of Hamiltonian matrices.
\subsection{Preliminaries}
    We consider linear dynamical systems with state-space form 
\begin{equation}
    \dot{x} = A x + B u,
\end{equation}
where $x$ is the \emph{state} of the system and $u$ is a \emph{control} action. We compactly denote this system as $\Sigma \triangleq (A,B)$. 

In this paper, we consider two different outputs from the system:
\begin{equation}
    \begin{aligned}
        y_p &= C_p x + D_p u, \\
        y_r &= C_r x + D_r u. \\ 
    \end{aligned}
\end{equation}
Accordingly, we define the systems $\Sigma_p \triangleq (A,B,C_p, D_p)$ and $\Sigma_r \triangleq (A,B,C_r, D_r).$ Moreover, the transfer function from $u$ to $y_p$ is
\begin{equation}
    G_p(s) \triangleq C_p(sI-A)^{-1}B + D_p,
\end{equation}
and similarly the transfer function from $u$ to $y_r$ is
\begin{equation}
    G_r(s) \triangleq C_r(sI-A)^{-1}B + D_r.
\end{equation}
\begin{remark}
    Typically, the signal $y_p$ is a measure of performance of the control system, where a smaller value of $y_p$ is desirable.
    The signal $y_r$ is typically the residual from an anomaly detector (a larger value of $y_r$ signify an anomaly.)
\end{remark}

The $\mathscr{L}_{2}$-norm, denoted $\|\cdot\|_{\mathscr{L}_2}$, of a signal $y$ is defined as 
\begin{equation}
    \|y\|_{\mathscr{L}_2} \triangleq \sqrt{\int_{0}^{\infty} \|y(t)\|^2_2 dt}.
\end{equation}

We consider, in particular, input signals $u$ that are in the \textit{extended} $\mathscr{L}_2$ space, denoted $\mathscr{L}_{2e}$, which is the set of all signals which has finite energy for all possible truncations: 
\begin{equation*}
    \mathscr{L}_{2e} \triangleq \left\{u: R \to R^{n_a} : \int_{0}^T \|u(t)\|^2_2 dt < \infty, \forall T < \infty \right\}.
\end{equation*}

So-called \emph{generalized singular values} will play an important role later on.
These are solutions to generalized eigenvalue problems, with the matrices being Gram matrices \cite{golub2013matrix}. Specifically, they are defined as follows. 

\begin{definition}[Generalized singular values]
The generalized singular values of the matrices $M$ and $N$, denoted ${\sigma}(M,N)$ is defined as 
\begin{equation*}
    {\sigma}(M,N) = \{\sigma > 0 : M^H M v = \sigma^2 N^H N v \text{ for } v\neq 0\},
\end{equation*}
where superscript $H$ denotes the conjugate transpose.
\end{definition}

Of special importance is the maximum singular value, which we denote $\bar{\sigma}$; that is, $\bar{\sigma}(M,N) \triangleq \max \sigma(M,N)$.

\subsection{The output-to-output gain}
\begin{definition}[The output-to-output gain]
    For a system $\Sigma$ with the outputs $y_p$ and $y_r$, the \textit{output-to-output gain} (OOG) from $y_r$ to $y_p$, denoted $\|\Sigma\|_{y_p \leftarrow y_r}$, is defined as  
    \begin{equation}
        \begin{aligned}
            \|\Sigma\|_{y_p \leftarrow y_r} \triangleq\sup_{u\in \mathscr{L}_{2e}}\: &\|y_p\|^2_{\mathscr{L}_{2}}  \\
            \text{s.t. }\:\:& \|y_r\|^2_{\mathscr{L}_2} \leq 1. \\
                                                                 & x(0) = 0.
        \end{aligned}
    \end{equation}
\end{definition}

An intuition behind the OOG is that it gives a measure of how large the output $y_p$ can become given that we restrict the size of another output $y_r$. In the context of security, a larger value of $y_p$ signify that the damage to the system is large, while a larger $y_r$ makes an attack easier to detect. An interpretation of the OOG in a security context is, hence, as a measure of the maximum possible damage for undetected attacks. 

\begin{remark}[$H_{\infty}$ as a special case]
    \label{rem:Hinf}
The well-known $H_{\infty}$-norm of a system is a special case of the OOG when $y_r = u$ (i.e., when $C_r = 0$ and $D_r = I$.) 
\end{remark}

The OOG can be characterized through dissipative systems theory \cite{trentelman1991dissipation} by using the quadratic supply rate 
\begin{equation}
    \label{eq:supply-rate}
    \begin{aligned}
        w_{\gamma}(x,u) &\triangleq \|y_p\|_2^2 - \gamma \|y_r\|_2^2  \\
        &= \|C_p x + D_p u\|_2^2 - \gamma \|C_r x + D_r u\|_2^2  \\
                                         &= 
{
\begin{bmatrix}
 x \\ u  
\end{bmatrix}^T
\begin{bmatrix}
Q_\gamma & S_{\gamma} \\
S_\gamma^T & R_\gamma
\end{bmatrix}
\begin{bmatrix}
 x \\ u  
\end{bmatrix},
}
    \end{aligned}
\end{equation}
where 
      $Q_{\gamma} \triangleq C_p^T C_p - \gamma C_r^T C_r, S_{\gamma} \triangleq C_p^T  D_p - \gamma C_r^T D_r$ and  
    $R_{\gamma} \triangleq D_p^TD_p - \gamma D_r^T D_r^T.$

    Using this supply rate, the OOG can be computed by solving the following semi-definite program \cite{teixeira2021security}
\begin{equation}
    \label{eq:oog-sdp}
    \begin{aligned}
        \|\Sigma\|_{y_p \leftarrow y_r}& = \min_{P \succeq 0, \gamma > 0} \gamma \\
                                          &\text{subject to }
                                          {\footnotesize
                                              \begin{bmatrix}
                                                  A^T P + P A + Q_\gamma  & P B + S_\gamma \\
                                                  B^T P + S_\gamma^T &  R_\gamma 
                                              \end{bmatrix} 
                                          }
                                          \preceq 0.
    \end{aligned}
\end{equation}

The OOG is related to the transfer functions $G_p$ and $G_r$ through the following inequality
\begin{equation}
    \label{eq:oog-freq}
    \|\Sigma\|_{y_p \leftarrow y_r} \geq \max_{s: \text{Re}(s) \geq 0 }\bar{\sigma}(G_p(s),G_r(s)),
\end{equation}
Since the optimization is carried out over the entire right-hand side of the complex plane, it is not simple to use this inequality in practice. It is, however, a precursor to a more useful frequency relationship for the so-called \emph{cyclic} OOG, which is introduced next.  
\begin{remark}
    The inequality in \eqref{eq:oog-freq} comes from a frequency condition in \cite{willems1971least} that is necessary but not sufficient \cite{willems1974freq}. There are, however, special cases when this frequency condition is also sufficient, which makes \eqref{eq:oog-freq} an equality \cite{moylan1975freq,molinari1975freq}. Even if \eqref{eq:oog-freq} holds with equality, however, the optimization needs to be carried out over the entire right-hand side of the complex plane. In contrast, the cyclic OOG that is introduced next requires a search over only the imaginary axis.
\end{remark}

\subsection{The cyclic output-to-output gain}
An alternative version of the OOG with computational benefits, the \emph{cyclic} OOG (COOG), can be derived by restricting the set of signals such that $\lim_{t \to \infty} x(t) =0$.
\begin{definition}[The cyclic output-to-output gain]
    For a system $\Sigma$ with the outputs $y_p$ and $y_r$, the \textit{cyclic output-to-output gain} (COOG) from $y_r$ to $y_p$, denoted $\|\Sigma\|_{y_p \hookleftarrow y_r}$, is defined as  
\begin{equation}
    \begin{aligned}
        \|\Sigma\|_{y_p \hookleftarrow y_r} \triangleq\sup_{u \in \mathscr{L}_{2e}}\: &\|y_p\|^2_{\mathscr{L}_2}  \\
        \text{subject to }& \|y_r\|^2_{\mathscr{L}_2} \leq 1, \\ 
                    & \:x(0) = 0,\:\: \lim_{t \to \infty} x(t) = 0. 
    \end{aligned}
\end{equation}
\end{definition}

Similar to how dissipativity was used for the OOG to derive the SDP \eqref{eq:oog-sdp}, one can use \textit{cyclodissipativity} \cite{moylan2014dissipative} to derive a SDP that computes the COOG:
\begin{equation}
    \label{eq:coog-sdp}
    \begin{aligned}
        \|\Sigma\|&_{y_p \hookleftarrow y_r} = \min_{P , \gamma > 0} \gamma \\
                                          &\text{subject to }
                                          {\footnotesize
                                              \begin{bmatrix}
                                                  A^T P + P A + Q_\gamma  & P B + S_\gamma \\
                                                  B^T P + S_\gamma^T &  R_\gamma 
                                              \end{bmatrix} 
                                          }
                                          \preceq 0,
    \end{aligned}
\end{equation}
which is identical to \eqref{eq:oog-sdp}, except that the constraint $P \succeq 0$ has been dropped.

The COOG has the frequency characterization
\begin{equation}
    \label{eq:coog-freq}
    \|\Sigma\|_{y_p \hookleftarrow y_r} = \max_{\omega} \bar{\sigma}(G_p(i \omega), G_r(i \omega)),
\end{equation}
where, in contrast to \eqref{eq:oog-freq}, the maximization is carried out over only the imaginary axis. Moreover, instead of the inequality in \eqref{eq:oog-freq}, the frequency condition in \eqref{eq:coog-freq} holds with equality. Hence, \eqref{eq:coog-freq} can be used to compute the COOG, which we exploit in the proposed method. 

As is shown in \cite{willems1971least}, the solutions of \eqref{eq:coog-sdp} are related to the solution $P$ of the ARE
\begin{equation}
    \label{eq:ARE}
    A^T P + P A +(P B + S_{\gamma}) R^{-1}_{\gamma} (B^T P + S^T_{\gamma}) + Q_{\gamma} = 0.
\end{equation}
This ARE exists iff $R_\gamma$ is nonsingular, and is used in \cite{boyd1989bisection,bruinsma1990fast} to efficiently compute the $H_{\infty}$-norm by relating the ARE to a Hamiltonian matrix \cite{willems1971least,laub1979schur}.
Such Hamiltonian-based methods are, however, not directly applicable to computing the COOG since the COOG might, in contrast to the $H_{\infty}$-norm, yield an $R_{\gamma}$ that is singular.  For example, the important special case of $D_p=0$ and $D_r = 0$ yields $R_\gamma = 0$, which is trivially singular. In this case, there is no corresponding ARE and, consequently, no corresponding Hamiltonian matrix. To be able to leverage efficient Hamiltonian-based methods for computing the COOG, we therefore introduce \emph{regularized} versions of the COOG to ensure that a corresponding ARE always exists. 

\subsection{The regularized cyclic output-to-output gain}
To ensure that the COOG always have a corresponding ARE of the form \eqref{eq:ARE}, we introduce a regularized version of it:
\begin{definition}[The regularized cyclic output-to-output gain]
    \label{def:roog}
    For a system $\Sigma$ with the outputs $y_p$ and $y_r$, the \textit{cyclic regularized output-to-output gain} (RCOOG) from $y_r$ to $y_p$ with regularization $\epsilon >0$, denoted $\|\Sigma\|_{y_p \overset{\epsilon}{\hookleftarrow} y_r}$, is defined as  
\begin{equation}
    \label{eq:roog}
    \begin{aligned}
        \|\Sigma\|_{y_p \overset{\epsilon}{\hookleftarrow} y_r} \triangleq\sup_{u\in \mathscr{L}_{2e}}\: &\|y_p\|^2_{\mathscr{L}_2}  \\
        \text{subject to } & \|y_r\|^2_{\mathscr{L}_2} + \epsilon \|a\|^2_{\mathscr{L}_2} \leq 1. \\
                           &x(0) = 0, \quad \lim_{t \to \infty} x(t) = 0.
    \end{aligned}
\end{equation}
\end{definition}
Specifically, the COOG is recovered from the RCOOG when $\epsilon \to 0$. For more insight/motivation behind this regularization, see  
Remark 26 in \cite{willems1971least}. 

The RCOOG can be seen as the COOG of an augmented system with output signal $\tilde{y}_r= \left(\begin{smallmatrix}
        y_r \\ \sqrt{\epsilon} u 
\end{smallmatrix}\right)$, which gives the equality $\|\Sigma\|_{y_p \overset{\epsilon}{\hookleftarrow} y_r} = \|\Sigma\|_{y_p  \hookleftarrow \tilde{y}_r}.$ This perspective gives a straightforward way to extend previous results. Semi-definite programs for computing the RCOOG are, for example, obtained by replacing $R_{\gamma}$ with $R_{\gamma} - \gamma \epsilon I$, which highlights how the regularization allows us to modify $R_{\gamma}$ such that it becomes non-singular and, in turn, such that the ARE in \eqref{eq:ARE} exists.
Explicitly, the RCOOG can be computed through the SDP
\begin{equation}
    \label{eq:rcoog-sdp}
    \begin{aligned}
        \|\Sigma\|&_{y_p \overset{\epsilon}{\hookleftarrow} y_r} = \min_{P , \gamma > 0} \gamma \\
                                          &\text{subject to }
                                          {\footnotesize
                                              \begin{bmatrix}
                                                  A^T P + P A + Q_\gamma  & P B + S_\gamma \\
                                                  B^T P + S_\gamma^T &  R_\gamma -\gamma \epsilon I 
                                              \end{bmatrix} 
                                          }
                                          \preceq 0.
    \end{aligned}
\end{equation}
Moreover, the frequency relationship in \eqref{eq:coog-freq} is extended to the RCOOG by replacing $\bar{\sigma}(G_p,G_r)$ with $\bar{\sigma}(G_p,\left(\begin{smallmatrix} G_r \\ \sqrt{\epsilon} I \end{smallmatrix}\right))$, resulting in the frequency characterization
\begin{equation}
    \label{eq:rcoog-freq}
    \|\Sigma\|_{y_p \overset{\epsilon}{\hookleftarrow} y_r} = \max_{\omega} \bar{\sigma}(G_p(i \omega), \left(\begin{smallmatrix} G_r(i \omega)\\ \sqrt{\epsilon} I \end{smallmatrix}\right))
\end{equation}
for the RCOOG. We will use this characterization in the proposed method in the upcoming section.

\subsection{RCOOG vs COOG}
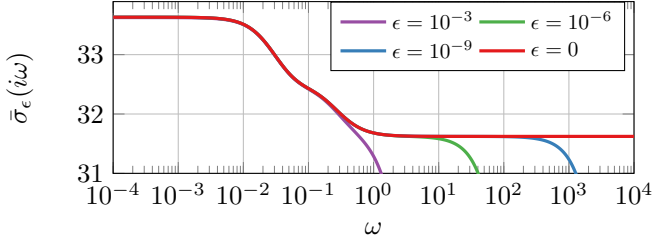
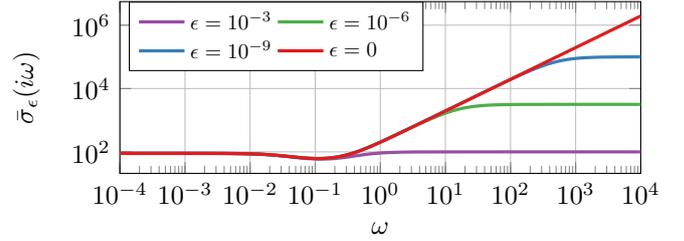
\begin{figure*}
    \centering
    \begin{subfigure}[b]{0.45\textwidth}
        \input{figs/reg_ex7.tex}
        \caption{Physical attack to the first tank of the tank system.}
        \label{fig:a}
    \end{subfigure}
    \qquad
    \begin{subfigure}[b]{0.45\textwidth}
        \input{figs/reg_ex3.tex}
        \caption{Actuator attack to the second tank of the tank system.}
        \label{fig:b}
    \end{subfigure}
    \caption{The maximum generalized singular values $\bar{\sigma}_\epsilon(i \omega) \triangleq \bar{\sigma}(G_p(i \omega), \left(\begin{smallmatrix} G_r(i \omega)\\ \sqrt{\epsilon} I \end{smallmatrix}\right))$ for the quadruple tank system for varying regularization $\epsilon$ and frequency $\omega$. Two different types of attacks from \cite{teixeira2021security} are considered. The main takeaway is that the regularization results in dampening of $\bar{\sigma}$ for higher frequencies.} 
    \label{fig:reg-example}
\end{figure*}
To show how the RCOOG qualitatively differs from the COOG, we consider a quadruple tank system under some attack scenarios considered in \cite{teixeira2021security}; we use the system matrices and controllers reported therein. We consider two attack scenarios: (a) a physical attack on one of the tanks; (b) an actuator attack on one of the tanks. For these systems, the resulting generalized singular values $\bar{\sigma}_\epsilon(i \omega) \triangleq \bar{\sigma}(G_p(i \omega), \left(\begin{smallmatrix} G_r(i \omega)\\ \sqrt{\epsilon} I \end{smallmatrix}\right))$ for different regularizations $\epsilon$ and different frequencies $\omega$ are shown in Figure~\ref{fig:reg-example}. In both cases, the main effect of the regularization is clear: $\bar{\sigma}_{\epsilon}( i\omega)$ is dampened for higher frequencies (the larger $\epsilon$ is, more lower frequencies are dampened.) Recall from \eqref{eq:rcoog-freq} that the RCOOG is the maximum of the generalized singular values taken over the frequencies. Moreover, $\epsilon = 0$ corresponds to the COOG. For attack scenario (a), the RCOOG is equal to COOG for the considered regularizations, since the maximum is obtained at lower frequencies (which is unaffected by the regularization.) This is a common case in practice, where the attack signals that can be applied contain lower frequencies. In attack scenario (b), the COOG ($\epsilon = 0$) is unbounded (as discussed in detail \cite{teixeira2021security}). In contrast, any choice of $\epsilon > 0$ makes the RCOOG bounded. From a practical perspective, this means that using the RCOOG instead of COOG can avoid singularities and unboundedness, which is advantageous when the metric is, for example, used for synthesis.

\section{Efficient Computation of the RCOOG}
\label{sec:comp}
For computing the RCOOG, we will exploit its frequency characterization in \eqref{eq:rcoog-freq}. A naive way of doing so would be to directly search over frequencies $\omega$, compute the corresponding $\bar{\sigma}$ for each frequency, and then pick the maximum value. We present a more efficient method that uses a relationship between $\bar{\sigma}$ and eigenvalues of a Hamiltonian matrix, which is the main contribution of this paper. Specifically, the method extends the method in \cite{bruinsma1990fast} that computes the $H_\infty$-norm, by considering generalized singular values rather than regular singular values. 

First, we derive the relationship between $\bar{\sigma}$ and the eigenvalues of a Hamiltonian matrix in Section~\ref{ssec:hamiltonian}, which we then use in Section~\ref{ssec:comp-alg} to propose Algorithm~\ref{alg:main} that efficiently computes the RCOOG. 

\subsection{Computing $\bar{\sigma}$ using a Hamiltonian matrix}
\label{ssec:hamiltonian}
As mentioned, a direct search over frequencies in \eqref{eq:coog-freq} is not an efficient way of computing the RCOOG. Instead, we will use an  ``inverse search'' over frequencies. That is, we will find all frequencies $\omega$ such that   
$\bar{\sigma}(G_p(i\omega),\left(\begin{smallmatrix}
    G_r(i\omega) \\ \sqrt{\epsilon} I 
\end{smallmatrix}\right)) = \gamma$ for any constant $\gamma$. The following theorem enables such an ``inverse search'', and is the main theoretical contribution of the paper. 
\begin{theorem}
    \label{th:hamiltonian}
    Assume that $A$ does not have any eigenvalues on the imaginary axis.
    Moreover, let   $ \mathcal{D} \triangleq (\gamma D_r^T D_r+\gamma \epsilon I - D_p^T D_p)$ and let $K = \mathcal{D}^{-1}(D_p^T C_p-\gamma D_r^T C_r)$. Then for all $\omega \in \mathbb{R}$,
    the matrices 
$(G_p(i\omega),\left(\begin{smallmatrix}
    G_r(i\omega) \\ \sqrt{\epsilon} I 
\end{smallmatrix}\right))$
    have a generalized singular value $\gamma$ iff the Hamiltonian matrix
    \begin{equation*}
        M_{\gamma} \triangleq 
        \footnotesize
        \begin{bmatrix}
            A+BK & - B \mathcal{D}^{-1} B^T \\
            -\gamma C_r^T (C_r+D_r K) + C_p^T (C_p+D_p K)  & -(A+BK)^T 
        \end{bmatrix}.
    \end{equation*}
 has an imaginary eigenvalue $i \omega$.
\end{theorem}
\begin{proof} The proof is given in the Appendix.
\end{proof}

\begin{remark}
    For the important special case when $D_p =0$ and $ D_r =0$, the Hamiltonian matrix $M_{\gamma}$ simplifies to  
    \begin{equation}
        M_{\gamma} =
        \begin{bmatrix}
            A & -\frac{1}{\epsilon \gamma} B B^T \\
            -\gamma C_r^T C_r + C_p^T C_p & -A^T 
        \end{bmatrix}.
    \end{equation}
\end{remark}

For $M_{\gamma}$ to be well-defined, the inverse of $\mathcal{D}$ needs to exist. This can be ensured by suitably selecting the regularization parameter $\epsilon$ (recall that the RCOOG $\rightarrow$ COOG when $\epsilon$ is made arbitrarily small.) 
In fact, as the following lemma shows, there always exists an arbitrarily small $\epsilon$ that makes $M_{\gamma}$ well-defined, in the sense that $\mathcal{D}^{-1}$ exists. 
\begin{lemma}
    For any $\gamma > 0$, there exist an arbitrarily small $\epsilon$ such that $M_{\gamma}$ is well-defined.  
\end{lemma}
\begin{proof}
    $M_{\gamma}$ is well-defined if the matrix $\mathcal{D} \triangleq \gamma D_r^T D_r - D^T_p D_p + \gamma \epsilon I$ is invertible. Let $\Lambda$ be the spectrum of $\gamma D_r^T D_r - D^T_p D_p$. If we select $\epsilon$ such that $-\gamma \epsilon \notin \Lambda$ we get that $0$ is not an eigenvalue of $\mathcal{D}$, which makes $\mathcal{D}$ invertible. Since $\Lambda$ is finite and countable, we can select $\epsilon$ arbitrarily small while ensuring that $-\gamma \epsilon \notin \Lambda$. (Note that if $0 \notin \Lambda$, we can select $\epsilon = 0$.)    
\end{proof}

We make the relationship between $M_{\gamma}$ and the RCOOG more explicit in the following corollary, which follows directly from Theorem~\ref{th:hamiltonian}, the frequency relationship in \eqref{eq:rcoog-freq} and the continuity of $\bar{\sigma}$.
\begin{corollary}
    \label{cor:ham-rcoog}
    Assume that $A$ is stable. Then the RCOOG $\|\Sigma\|_{y_p \overset{\epsilon}{\hookleftarrow} y_r} < \gamma$ iff $M_\gamma$ has no imaginary eigenvalue. 
\end{corollary}

\subsection{An efficient algorithm for computing the RCOOG} 
\label{ssec:comp-alg}

Next, we use the ``inverse search'' over frequency enabled by Theorem~\ref{th:hamiltonian} to propose a method for computing the RCOOG, given in Algorithm \ref{alg:main}. The method extends of the method in \cite{bruinsma1990fast}, and requires computing generalized, rather than regular, singular values. 
On a high level, Algorithm~\ref{alg:main} iteratively refines a lower bound on the RCOOG by computing imaginary eigenvalues of a the Hamiltonian matrices $M_\gamma$ introduced in Theorem~\ref{th:hamiltonian}.
The motivation for each step of the algorithm is as follows. 

\begin{algorithm}
  \caption{Computing the RCOOG using Hamiltonian matrices.}
  \label{alg:main}
  \begin{algorithmic}[1]
      \Require $\Sigma_p$, $\Sigma_r$, regularization $\epsilon$, bound $\underline{\gamma}$, tolerance $\epsilon_{\gamma}$ 
      \Ensure The RCOOG $\|\Sigma\|_{y_p \overset{\epsilon}{\hookleftarrow} y_r}$
      \State $\gamma \leftarrow (1+2 \epsilon_{\gamma}) \underline{\gamma}$ \label{st:gamma-up}
    \State $\Lambda_\text{Im} \leftarrow$ get all imaginary eigenvalues of $M_{\gamma}$ \label{st:query}
    \If{$\Lambda_{\text{Im}} = \emptyset$} \label{st:upperbound}
    \State \textbf{return} $\frac{1}{2}(\underline{\gamma} + \gamma)$ \label{st:return}
    \Else
    \State $\{\omega_i\}_{i=1}^N \leftarrow $  sorted magnitudes of elements in $\Lambda_{\text{Im}}$. \label{st:extract}
    \For {$i \in \{1,\dots,N-1\}$} \label{st:for}
    \State $\bar{\omega} \leftarrow \frac{1}{2}(\omega_i + \omega_{i+1})$
    \State $\underline{\gamma} \leftarrow \max\left\{\underline{\gamma},\:\: 
        \bar{\sigma}(G_p(i \bar{\omega}),\left(\begin{smallmatrix}
                G_r(i \bar{\omega}) \\ \sqrt{\epsilon} I 
    \end{smallmatrix}\right))\right\}$ \label{st:up}
    \EndFor
    \EndIf
    \State \textbf{goto} Step 1
  \end{algorithmic}
\end{algorithm}

\textbf{(Step~\ref{st:gamma-up})} Each iteration starts with a lower bound $\underline{\gamma}$ of the RCOOG . We increase this lower bound with a relative tolerance $\epsilon_{\gamma}$ to get $\gamma$. 

\textbf{(Step~\ref{st:query})}  Frequencies for which $\bar{\sigma} = \gamma$ are determined by computing the imaginary eigenvalues of $M_{\gamma}$, motivated by Theorem~\ref{th:hamiltonian}. 

\textbf{(Step~\ref{st:upperbound}-\ref{st:return})} If $M_{\gamma}$ has no imaginary eigenvalues, Corollary~\ref{cor:ham-rcoog} implies that $\|\Sigma\|_{y_p \overset{\epsilon}{\hookleftarrow} y_r}< \gamma = (1+2\epsilon_{\gamma})\underline{\gamma}$. Moreover, we already know that $\underline{\gamma}$ is a lower bound, so we have that 
\begin{equation}
    \underline{\gamma} \leq \|\Sigma\|_{\overset{\epsilon}{y_p \hookleftarrow y_r}} < (1+2 \epsilon_r) \underline{\gamma}.
\end{equation}
The algorithm return the midpoint of this interval, $\frac{1}{2}(\underline{\gamma}+\gamma)$, which has a maximum relative error $\epsilon_\gamma$ from the true RCOOG.  

\textbf{(Step~\ref{st:extract})} If $M_{\gamma}$ has imaginary eigenvalues, we extract the corresponding frequencies for which $\bar{\sigma} = \gamma$. 

\textbf{(Step~\ref{st:for}-\ref{st:up})} For the extracted frequencies, we evaluate $\bar{\sigma}$ for the midpoints between them, and update our lower bound $\underline{\gamma}$ with the maximum of these evaluations.

Algorithm~\ref{alg:main} requires an initial lower bound $\underline{\gamma}$ on the RCOOG. We initialize this lower bound with 
\begin{equation}
    \bar{\gamma} \leftarrow \max\left\{
        \bar{\sigma}(G_p(0),\left(\begin{smallmatrix}G_r(0) \\ \sqrt{\epsilon}I \end{smallmatrix}\right)),\:\:
        \bar{\sigma}(D_p,\left(\begin{smallmatrix} D_r \\ \sqrt{\epsilon}I \end{smallmatrix}\right))
\right\}, 
\end{equation}
which corresponds to the maximum generalized singular values for frequencies $\omega = 0$ and $\omega \to \infty$.

\section{Numerical Experiments}
\label{sec:result}
Next, we investigate the efficiency of computing the RCOOG through a Hamiltonian matrix by using Algorithm~\ref{alg:main}. We do this by comparing Algorithm~\ref{alg:main} with the customary way of computing the RCOOG through an SDP of the form \eqref{eq:rcoog-sdp}. In Section \ref{ssec:random-ex}, we compare the methods on randomly generated systems. In Section \ref{ssec:ncs-ex}, we compare the methods on network control system, which allows for a comparison with the scalable way of computing OOGs proposed in \cite{nguyen2024scalable}. With the experiments\footnote{Code to reproduce all of the numerical results is available at \url{https://github.com/darnstrom/oog-compute}.}
, we aim to support the following claims: 

\begin{enumerate}
    \item Algorithm \ref{alg:main} is more \textbf{efficient} than solving an SDP. 
    \item Algorithm \ref{alg:main} is more \textbf{scalable} than solving an SDP.
    \item Algorithm \ref{alg:main} is more \textbf{reliable} than solving an SDP. 
\end{enumerate}

\subsection{Randomly generated systems}
\label{ssec:random-ex}
We generate stable system with varying state dimension $n_x$. For each system, the number of inputs and outputs are each $\frac{n_x}{5}$ (for both $\Sigma_p$ and $\Sigma_r$.) The elements of $B,C_p,D_p,C_r,$ and $D_r$ are all independently drawn from a normal distribution with mean 0 and variance~1. The matrix $A$ is generated as $A = T^{-1} \Lambda T$,  where $T$ is a similarity transform with elements randomly drawn from a normal distribution with mean 0 and variance~1, and $\Lambda$ is a diagonal matrix that contains stable poles.

For the generated systems, we compute the RCOOG with the regularization $\epsilon = 10^{-8}$ using Algorithm~\ref{alg:main}, and by solving an SDP of the form \eqref{eq:rcoog-sdp}. The SDPs are solved using the state-of-the-art conic solver MOSEK \cite{mosek}. For each state dimension $n_x$, we randomly generate 100 systems and record the execution times and RCOOGs. For MOSEK, we use the solve time as execution time (i.e., we disregard the setup time.) 
To get a ground truth of the RCOOG for each system, we grid 10000 frequencies on a logarithmic scale and compute the maximum of $\bar{\sigma}(G_p(i\omega),\left(\begin{smallmatrix} G_r(i \omega) \\ \sqrt{\epsilon} I \end{smallmatrix}\right))$ over these frequencies.

The upper part of Figure \ref{fig:roog-time} shows the average and best/worst-case execution times for computing the RCOOG. The lower part of Figure \ref{fig:roog-time} shows accuracy of the computed RCOOGS, with a result classified as correct if it is better or within 5\% of the value obtained by gridding. 

Figure~\ref{fig:roog-time} illustrates that Algorithm~\ref{alg:main} computes the RCOOG significantly faster (one to three orders of magnitude) than solving and SDP. The speedup also improves for larger systems, which shows that Algorithm~\ref{alg:main} is more scalable. For accuracy, Algorithm~\ref{alg:main} always computes the correct value of the RCOOG, while the accuracy when solving the SDP deteriorates as the system dimension grows.
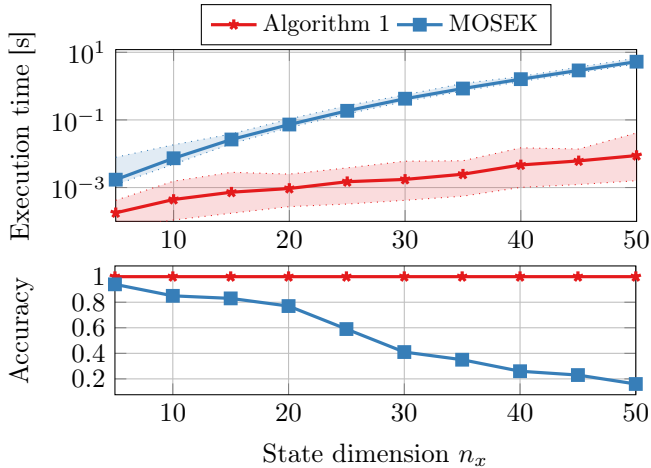
\begin{figure}
  \centering
  \input{figs/times.tex}
  \input{figs/accuracy.tex}

  \caption{Time for computing the RCOOG ($\epsilon = 10^{-8}$) for systems of different dimension $n_x$. For each state dimension, 100 stable systems were randomly generated, with $\frac{n_x}{5}$ inputs and $\frac{n_x}{5}$ outputs (for both $\Sigma_r$ and $\Sigma_p$.) The solid lines in the upper figure show the average times for computing the RCOOG, and the dotted lines shows best/worst-case times. The lower figure show the fraction of correctly computed RCOOGs.}
  \label{fig:roog-time}
\end{figure}

\subsection{Positive networked control systems}
\label{ssec:ncs-ex}
Next, we consider networked control systems similar to what is considered in \cite{nguyen2024scalable}. For these systems, $A = -L$, where $L$ is an in degree Laplacian matrix for a strongly connected directed graph with $N$ nodes. Given a topology of the graph, the edge weights are randomly drawn from a uniform distribution with values between 0.8 and 1.2. In accordance with \cite{nguyen2024scalable}, we also add a self-loop with weight 1 to each node to regularize the systems such that some of the results in \cite{nguyen2024scalable} applies. The topology of the graph is generated by uniform sampling of all simple directed graphs with $N$ nodes and $N$ edges, followed by computing the connected components of this graph, and then adding a edge between a randomly selected nodes in each connected component.

The performance signal $y_p$ is the sum of all the states ($C_p = [1 \:\: \dots \:\: 1]$ and $D_p = 0$.) The residual signal $y_r$ is obtained by randomly selecting 2\% of the nodes in the graph; that is 
\begin{equation}
  C_r = \left[\begin{smallmatrix}
      e^T_{m_1} \\ \vdots \\ e^T_{m_{d}} 
\end{smallmatrix}\right], \quad D_r = 0,
\end{equation}
where $e_i$ is the $i$th  unit vector, the number of residuals are $d = N/50$ , and $\{m_1,\dots,m_d\}$ are uniformly drawn from $\{1,\dots, N\}$ without replacement.

We compare (i) Algorithm~\ref{alg:main}; (ii) solving the SDP in \eqref{eq:rcoog-sdp} with a dense $P$; (iii) solving the SDP in \eqref{eq:rcoog-sdp} with a diagonal $P$ (proposed in \cite{nguyen2024scalable}). When solving the SDPs, MOSEK \cite{mosek} performed the best when $P$ was dense, while Clarabel \cite{goulart2024clarabel} performed the best when $P$ was diagonal (since the chordal decomposition available in Clarabel can be exploited). We therefore only report the results for MOSEK when $P$ is dense, and for Clarabel when $P$ is diagonal. 
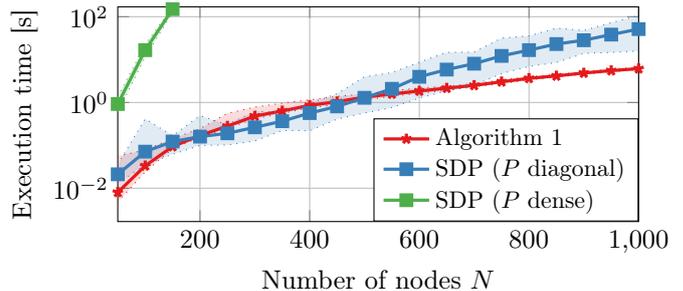
\begin{figure}
  \centering
  \input{figs/ncs.tex}
  \caption{Execution time for computing RCOOGs for networked control systems. SDPs when $P$ is dense is solved with MOSEK, while SDPs when $P$ is diagonal is solved with Clarabel to exploit its chordal decomposition. For a diagonal $P$, the relative error of the computed RCOOG was 7.2\% on average and over 900\% in the worst case.}
  \label{fig:ncs-time}
\end{figure}

Figure~\ref{fig:ncs-time} shows the average and best/worst-case execution times for computing the RCOOG with the regularization\footnote{A larger regularization compared with Section~\ref{ssec:random-ex} was required to avoid numerical errors when solving the SDPs. Algorithm~\ref{alg:main} worked fine with even a smaller regularization of $\epsilon = 10^{-8}$.} $\epsilon = 10^{-5}$.
In accordance with \cite{nguyen2024scalable}, we see that a diagonal $P$ scales significantly better than a dense $P$. Still, Algorithm~\ref{alg:main} scales even better than using a diagonal $P$. 

Regarding accuracy, using a diagonal $P$ does not guarantee that the correct RCOOG is computed (although, as is shown in \cite{nguyen2024scalable}, it often yields values close to the correct RCOOG.) In contrast, Algorithm~\ref{alg:main} always gives the correct value. For our experiments, the relative error between the RCOOG computed with a diagonal $P$ and with Algorithm~\ref{alg:main} was on average 7.2\%, and in the worst case over 900\%.

In our experiments, we also observe that using a diagonal $P$ works well when the graphs are relatively sparse, since then the chordal decomposition in Clarabel is very effective. But execution times quickly increase as the graphs get denser. For example, if the lower bound on the number of edges is increased from 500 to 1500 edges for $N=500$, the average execution time for solving \eqref{eq:rcoog-sdp} when $P$ is diagonal increases from 1.3 seconds to to $\approx$400 seconds. The execution time for the proposed Algorithm~\ref{alg:main} does, in contrast, remain unchanged at around 1 second, since Algorithm~\ref{alg:main} does not exploit sparsity.

\section{Conclusion}
\label{sec:conc}
We have proposed an efficient method for computing the output-to-output gain. The method is based on computing eigenvalues of Hamiltonian matrices and extends the method in \cite{bruinsma1990fast}. Numerical experiments show that the proposed method is more efficient, scalable, and reliable than methods based on solving semidefinite programs.  

While this paper only consider continuous-time systems, the approach can readily be extended to discrete-time systems, similar to how \cite{bongers1991norm} extends \cite{bruinsma1990fast}. Future work includes using the new way of efficiently compute the output-to-output in nonsmooth optimization, similar to how \cite{bruinsma1990fast} is used in \cite{apkarian2006nonsmooth}.

\appendix
\section{}

In this appendix we prove Theorem~\ref{th:hamiltonian}. The proof follows the same structure as the proof of Theorem~1 in \cite{bruinsma1990fast}, but differ in that we consider generalized singular values rather than regular singular values.

\subsection*{State-space realizations of frequency responses}
To clean up the proof, we first introduce some notation and properties that relates frequency responses and state-space realizations.
Similar to \cite{francis1987course}, we let the frequency response for a state-space system $(A,B,C,D)$ be denoted
\begin{equation}
    [A,B,C,D] \triangleq D + C(i \omega I-A)^{-1} B.
\end{equation}

The following properties can easily be shown \cite[pp.VIII]{francis1987course}: 
\begin{property}[Inverse of frequency response]
    \label{prop:inverse}
\begin{equation}
    \small
      [A,B,C,D]^{-1} = [A-C D^{-1} B, B D^{-1}, -D^{-1} C, D^{-1}] \\
\end{equation}
\end{property}
\begin{property}[Adjoint of frequency response]
    \label{prop:adjoint}
\begin{equation}
      [A,B,C,D]^H = [-A^T,-C^T, B^T, D^T] \\
\end{equation}
\end{property}

\begin{property}[Multiplication of frequency responses]
    \label{prop:mult}
\begin{equation}
    \begin{aligned}
        & [A_1,B_1,C_1,D_1] \cdot [A_2,B_2,C_2,D_2] \\
        &= \left[\begin{bmatrix}
                A_2 & 0 \\
                B_1 C_1 & A_1 
        \end{bmatrix}, \begin{bmatrix}
         B_2 \\ B_1 D_2 
        \end{bmatrix}, 
        \begin{bmatrix}
            D_1 C_2 & C_1 
        \end{bmatrix},
        D_1 D_2
    \right].
    \end{aligned}
\end{equation}
\vspace{0pt}
\end{property}
\begin{property}[Addition of frequency responses]
    \label{prop:add}
\begin{equation}
    \begin{aligned}
        & [A_1,B_1,C_1,D_1] + [A_2,B_2,C_2,D_2] \\
        &= \left[\begin{bmatrix}
                A_1 & 0 \\
                0 & A_2 
        \end{bmatrix}, \begin{bmatrix}
         B_1 \\ B_2
        \end{bmatrix}, 
        \begin{bmatrix}
            C_1 & C_2 
        \end{bmatrix},
        D_1 + D_2
    \right].
    \end{aligned}
\end{equation}
\vspace{0pt}
\end{property}

\subsection*{Hamiltonian system}

Before proving Theorem~\ref{th:hamiltonian}, we provide some partial results in the form of lemmas that we will use in the proof. In particular, these lemmas establish properties of the transfer function 
\begin{equation}
    \label{eq:Gamma}
\Gamma(i \omega) \triangleq \gamma  G_r^H(i \omega) G_r(i \omega) - G_p^H(i \omega) G_p(i \omega) + \gamma \epsilon I.
\end{equation}  
The following lemma relates $\Gamma$ and the Hamiltonian matrix $M_{\gamma}$ that is used in Theorem~\ref{th:hamiltonian}.
\begin{lemma}
    \label{lem:Gamma-Hamiltonian}
    There exists a state-space realization $(\tilde{A},\tilde{B},\tilde{C},\tilde{D})$ for $\Gamma(i \omega)^{-1}$ (defined in \eqref{eq:Gamma}) with $\tilde{A} = M_{\gamma}$. 
\end{lemma}
\begin{proof}
    First, we use Property~\ref{prop:adjoint} and Property~\ref{prop:mult} to conclude that $\gamma G_r(i \omega)^H G_r(i \omega)$ has the state space realization   
    \begin{equation*}
        \small
        \left(
      \begin{bmatrix}
          A  & 0 \\
          -\gamma C_r^T C_r & -A^T
      \end{bmatrix},
      \begin{bmatrix}
          B \\  
          -\gamma C_r^T D_r
      \end{bmatrix},
      \begin{bmatrix}
          \gamma D_r^T C_r & B^T
      \end{bmatrix},
      \gamma D_r ^T D_r
  \right)
    \end{equation*}
    Similarly, $G_p^H G_p$  has the state-space realization
    \begin{equation*}
        \small
        \left(
      \begin{bmatrix}
          A  & 0 \\
          -C_p^T C_p & -A^T
      \end{bmatrix},
      \begin{bmatrix}
          B \\  
          -C_p^T D_p
      \end{bmatrix},
      \begin{bmatrix}
          D_p^T C_p & B^T
      \end{bmatrix},
      D_p ^T D_p
  \right).
    \end{equation*}

    Property~\ref{prop:add} and state condensation then gives that $\Gamma(i \omega) \triangleq \gamma  G_r^H(i\omega) G_r(i \omega) - G_p^H(i \omega ) G_p(i \omega) + \gamma \epsilon I$ has the state-space realization $(\mathcal{A},\mathcal{B},\mathcal{C},\mathcal{D})$ with 
\begin{equation*}
    \small
    \begin{aligned}
        \mathcal{A} &\triangleq 
    \begin{bmatrix}
        A  & 0 \\
        -\gamma C_r^T C_r + C_p^T C_p & -A^T
    \end{bmatrix},\:\:
          \mathcal{B}  \triangleq 
        \begin{bmatrix}
            B \\  
            -\gamma C_r^T D_r + C_p^T D_p
        \end{bmatrix}, \\
        \mathcal{C} &\triangleq 
        \begin{bmatrix}
            \gamma D_r^T C_r-D_p^T C_p & B^T
        \end{bmatrix},\quad
          \mathcal{D} \triangleq \gamma D_r ^T D_r -D_p^T D_p + \gamma \epsilon I.
    \end{aligned}
\end{equation*}
Finally, Property~\ref{prop:inverse} results in $\Gamma(i \omega)^{-1}$ having a state-space realization $(\tilde{A},\tilde{B},\tilde{C},\tilde{D})$ with 
\begin{equation}
    \tilde{A} \triangleq \mathcal{A} - \mathcal{B} \mathcal{D}^{-1} \mathcal{C},
\end{equation}
with the second term
\begin{equation*}
    \small 
    \begin{aligned}
        \mathcal{B} \mathcal{D}^{-1} \mathcal{C}
        & = \begin{bmatrix}
            B \\  
            -\gamma C_r^T D_r + C_p^T D_p
        \end{bmatrix} \mathcal{D}^{-1} 
        \begin{bmatrix}
            \gamma D_r^T C_r-D_p^T C_p & B^T
        \end{bmatrix}\\
        & = \begin{bmatrix}
            -B K   & B \mathcal{D}^{-1} B^T \\
            (\gamma C_r^TD_r - C_p^T D_p)K & -K^T B^T 
        \end{bmatrix},
    \end{aligned}
\end{equation*}
where we have defined $K \triangleq \mathcal{D}^{-1}(-\gamma D_r^T C_r + D_p^T C_p)$.
Substracting $\mathcal{B} \mathcal{D}^{-1} \mathcal{C}$ from $\mathcal{A}$ gives 
    $\tilde{A} = M_{\gamma}$.
\end{proof}

With the relationship between $\Gamma^{-1}$ and $M_{\gamma}$ established, we show how the poles of $\Gamma^{-1}$ relates to $M_{\gamma}$. 

\begin{lemma}
    \label{lem:un}
    Let $(\tilde{A},\tilde{B},\tilde{C},\tilde{D})$ be the state-space realization of $\Gamma(i\omega)^{-1}$, and let $\lambda$ be an eigenvalue of $M_{\gamma}$. Then if $\pm \lambda$  is not an eigenvalue of $A$, then $\lambda$ is a pole of $\Gamma^{-1}$.
\end{lemma}
\begin{proof}
    First, by continuing from the end of the proof of Lemma~\ref{lem:Gamma-Hamiltonian}, we get from Property~\ref{prop:inverse} that
    \begin{equation}
        \tilde{B} = \begin{bmatrix}
            B \mathcal{D}^{-1} \\ K^T  
        \end{bmatrix},\quad
        \tilde{C} = [K \quad -\mathcal{D}^{-1} B^T].
    \end{equation}
    An eigenvalue $\lambda$ of $\tilde{A}$ (i.e., of $M_{\gamma}$) might not be a pole of $\Gamma^{-1}$ if there exists an uncontrollable or unobservable mode associated with $\lambda$.

If there exists uncontrollable modes for $\lambda$, there exists a vector $\left[\begin{smallmatrix} x_1  \\ x_2 \end{smallmatrix}\right] \neq 0$. Such that   $\left[\begin{smallmatrix} x_1  \\ x_2 \end{smallmatrix}\right]^T [\lambda I - M_{\gamma} \tilde{B}] = 0.$ That is,
    \begin{equation*}
      \begin{bmatrix}
       x_1 \\ x_2  
      \end{bmatrix}^T
        \begin{bmatrix}
            \lambda I - (A+BK) & B \mathcal{D}^{-1} B^T & B \mathcal{D}^{-1}\\
            -M_{21}  & \lambda I + (A+BK)^T  & K^T
        \end{bmatrix} = 0,
    \end{equation*}
    where we have defined $M_{21} \triangleq -\gamma C_r^T (C_r+D_r K) + C_p^T (C_p+D_p K)$. Written out, we have the equation system

    \begin{subequations}
      \begin{align}
          x_1^T (\lambda I - (A+BK)) - x_2^T M_{21} &= 0, \label{eq:unc-1} \\
          x_1^T B \mathcal{D}^{-1} B^T + x_2^T (\lambda I + (A+BK)^T) &= 0, \label{eq:unc-2} \\
          x_1^T B \mathcal{D}^{-1}+ x_2^T K^T &= 0 \label{eq:unc-3}.
      \end{align}
    \end{subequations}
    Combining \eqref{eq:unc-2} and \eqref{eq:unc-3} gives that $x_2^T(\lambda I + A^T) = 0$, which implies that that $x_2 = 0$ or $-\lambda$ is an eigenvalue of $A$. For the case $x_2=0$, combining \eqref{eq:unc-1} and \eqref{eq:unc-3} gives that $x_1^T(\lambda I - A)$, which implies that $x_1 = 0$ or $\lambda$ is an eigenvalue of $A$. So $(\tilde{A},\tilde{B}, \tilde{C}, \tilde{D})$ only have an uncontrollable mode associated with $\lambda$ if $\lambda$ or $-\lambda$ is an eigenvalue of $A$. 

    Similarly, there is an unobservable mode associated with $\lambda$ if there exists $\left[\begin{smallmatrix} x_1  \\ x_2 \end{smallmatrix}\right] \neq 0$ such that $\left[\begin{smallmatrix} \tilde{C}  \\ \lambda I - M_{\gamma} \end{smallmatrix}\right] \left[\begin{smallmatrix} x_1  \\ x_2 \end{smallmatrix}\right] = 0.$ Similar computations as above gives that $(\tilde{A},\tilde{B}, \tilde{C}, \tilde{D})$ can only have an unobservable mode associated with $\lambda$ if $\lambda$ or $-\lambda$ is an eigenvalue of $A$. 

    In conclusion, there can only be an unobservable or uncontrollable mode associated with $\lambda$ if $\lambda$ or $-\lambda$ is an eigenvalue of $A$, which in turn means that an eigenvalue $\lambda$ of $M_{\gamma}$ is a pole of $\Gamma^{-1}$ if $\pm \lambda$ is not an eigenvalue of $A$.
\end{proof}

Next, we relate the transmission zeros of $\Gamma(i \omega)$ (which relates to the poles of $\Gamma^{-1}$) to generalized singular values. 
\begin{lemma}
    \label{lem:trans-zero}
    The transfer function $\Gamma(s)$ has a transmission zero in $i \omega$ iff $\gamma$ is a generalized singular value to the matrices 
$(G_p(i\omega),\left(\begin{smallmatrix}
    G_r(i\omega) \\ \sqrt{\epsilon} I 
\end{smallmatrix}\right))$.
\end{lemma}

\begin{proof}
    If $i\omega$ is a transmission zero to $\Gamma$, we have that $\det \Gamma(i \omega) = 0$. Hence, $i\omega$ is a transmission zero to $\Gamma$ iff  
    \begin{equation*}
        \begin{aligned}
            \det \Gamma(i \omega)  &= \det \gamma G_r^H(i\omega) G_r(i \omega) - G_p^H(i\omega) G_p(i \omega) + \gamma \epsilon I \\ 
                                   &= \det 
                                   \gamma \begin{bmatrix}
                                    G_r(i\omega)  \\
                                    \sqrt\epsilon I 
                                   \end{bmatrix}^H 
                                   \begin{bmatrix}
                                    G_r(i\omega) \\
                                    \sqrt\epsilon I 
                                   \end{bmatrix} 
                                   - G_p^H(i\omega) G_p(i \omega) \\
                                   &= 0, 
        \end{aligned}
    \end{equation*}
where the last equality gives $\gamma \in \sigma
(G_p(i\omega),\left(\begin{smallmatrix}
    G_r(i\omega) \\ \sqrt{\epsilon} I 
\end{smallmatrix}\right)),$ which is the desired result. 
\end{proof}

\subsection*{Proof of Theorem~\ref{th:hamiltonian}}

\begin{proof}
    From Lemma~\ref{lem:Gamma-Hamiltonian}, poles of $\Gamma^{-1}$ is a subset of the eigenvalues of $M_{\gamma}$. Since, we assume that $A$ has no eigenvalues on the imaginary axis, it follows from Lemma~\ref{lem:un} that all purely imaginary eigenvalues of $M_{\gamma}$ are poles of $\Gamma^{-1}$. So $i\omega$ being an eigenvalue of $M_{\gamma} \Leftrightarrow i \omega$ is a pole to $\Gamma^{-1}$. Finally, since the imaginary poles of $\Gamma^{-1}$ are the transmission zeros of $\Gamma$, Lemma~\ref{lem:trans-zero} gives that $\gamma$ is a generalized singular value to the matrices $\left(G_p(i \omega), \left(\begin{smallmatrix} G_r (i \omega) \\ \sqrt{\epsilon} I \end{smallmatrix}\right)\right)$.
\end{proof}

%% file: figs/reg_ex7.tex
\begin{tikzpicture}[scale=1]
    \begin{axis}[
        ymin = 31, 
        xmin = 0.0001, xmax=10000,
        scale=1,
        xmode=log,
        xlabel={$\omega$},
    ylabel={$\bar{\sigma}_{\epsilon}(i\omega)$},
        legend style={at ={(0.7,1)},anchor=north}, ymajorgrids,yminorgrids,xmajorgrids,
        x post scale=1,
        y post scale=0.4,
        legend cell align={left},legend columns=2,
        legend style={nodes={scale=0.8, transform shape}},
        ]
        \addplot [set19c4,very thick] table [x={w}, y={em3}] {\exreg}; 
        \addplot [set19c3,very thick] table [x={w}, y={em6}] {\exreg}; 
        \addplot [set19c2,very thick] table [x={w}, y={em9}] {\exreg}; 
        \addplot [set19c1,very thick] table [x={w}, y={e0}] {\exreg}; 
        \legend{$\epsilon=10^{-3}$,$\epsilon=10^{-6}$,$\epsilon=10^{-9}$,$\epsilon=0$};
    \end{axis}
\end{tikzpicture}

%% file: figs/reg_ex3.tex
\begin{tikzpicture}[scale=1]
    \begin{axis}[
        xmin = 0.0001, xmax=10000,
        scale=1,
        xmode=log,
        ymode=log,
        xlabel={$\omega$},
        ylabel={$\bar{\sigma}_{\epsilon}(i\omega)$},
        legend style={at ={(0.3,1)},anchor=north}, ymajorgrids,yminorgrids,xmajorgrids,
        x post scale=1,
        y post scale=0.4,
        legend cell align={left},legend columns=2,
        legend style={nodes={scale=0.8, transform shape}},
        ]
        \addplot [set19c4,very thick] table [x={w}, y={em3}] {\exregunb}; 
        \addplot [set19c3,very thick] table [x={w}, y={em6}] {\exregunb}; 
        \addplot [set19c2,very thick] table [x={w}, y={em9}] {\exregunb}; 
        \addplot [set19c1,very thick] table [x={w}, y={e0}] {\exregunb}; 
        \legend{$\epsilon=10^{-3}$,$\epsilon=10^{-6}$,$\epsilon=10^{-9}$,$\epsilon=0$};

    \end{axis}
\end{tikzpicture}

%% file: figs/times.tex
\begin{tikzpicture}[scale=1]
    \begin{axis}[
        xmin = 5, xmax=50,
        scale=1,
        ymode=log,
        ymin=10^-4,ymax=1.2*10^1,
        ylabel={Execution time [s]},
        legend style={at ={(0.5,1.025)},anchor=south}, ymajorgrids,yminorgrids,xmajorgrids,
        x post scale=1,
        y post scale=0.4,
        legend cell align={left},legend columns=2,
        legend style={nodes={scale=0.9, transform shape}},
        ]
        \addplot [set19c1,very thick, mark=star] table [x={nx}, y={tavg}] {\hamiltonian}; 
        \addplot [set19c2,very thick, mark=square*] table [x={nx}, y={tavg}] {\mosek}; 

        \addplot [set19c1,dotted, name path=A] table [x={nx}, y={tmax}] {\hamiltonian}; 
        \addplot [set19c1,dotted, name path=B] table [x={nx}, y={tmin}] {\hamiltonian}; 
        \addplot [set19c1,fill opacity=0.15] fill between [of=A and B];

        \addplot [set19c2,dotted, name path=Al] table [x={nx}, y={tmax}] {\mosek}; 
        \addplot [set19c2,dotted, name path=Bl] table [x={nx}, y={tmin}] {\mosek}; 
        \addplot [set19c2,fill opacity=0.15] fill between [of=Al and Bl];
        \legend{Algorithm~\ref{alg:main} , MOSEK};

    \end{axis}
\end{tikzpicture}

%% file: figs/accuracy.tex
\begin{tikzpicture}[scale=1]
    \begin{axis}[
        xmin = 5, xmax=50,
        scale=1,
        xlabel={State dimension $n_x$},
        ylabel={Accuracy},
        ymajorgrids,yminorgrids,xmajorgrids,
        x post scale=1,
        y post scale=0.3,
        ]
        \addplot [set19c1,very thick, mark=star] table [x={nx}, y={accuracy}] {\hamiltonian}; 
        \addplot [set19c2,very thick, mark=square*] table [x={nx}, y={accuracy}] {\mosek}; 
    \end{axis}
\end{tikzpicture}

%% file: figs/ncs.tex
\begin{tikzpicture}[scale=1]
    \begin{axis}[
        xmin = 50, xmax=1000,
        scale=1,
        ymode=log,
        ymax=1.7*10^2,
        xlabel={Number of nodes $N$},
        ylabel={Execution time [s]},
        legend style={at ={(1.0,0)},anchor=south east}, ymajorgrids,yminorgrids,xmajorgrids,
        y post scale=0.5,
        legend cell align={left},legend columns=1,
        legend style={nodes={scale=0.9, transform shape}},
        ]
        \addplot [set19c1,very thick, mark=star] table [x={nx}, y={tavg}] {\ncshamiltonian}; 
        \addplot [set19c2,very thick, mark=square*] table [x={nx}, y={tavg}] {\ncsclarabel}; 
        \addplot [set19c3,very thick, mark=square*] table [x={nx}, y={tavg}] {\ncsmosek}; 

        \addplot [set19c1,dotted, name path=A] table [x={nx}, y={tmax}] {\ncshamiltonian}; 
        \addplot [set19c1,dotted, name path=B] table [x={nx}, y={tmin}] {\ncshamiltonian}; 
        \addplot [set19c1,fill opacity=0.15] fill between [of=A and B];

        \addplot [set19c2,dotted, name path=Al] table [x={nx}, y={tmax}] {\ncsclarabel}; 
        \addplot [set19c2,dotted, name path=Bl] table [x={nx}, y={tmin}] {\ncsclarabel}; 
        \addplot [set19c2,fill opacity=0.15] fill between [of=Al and Bl];

        \addplot [set19c3,dotted, name path=Am] table [x={nx}, y={tmax}] {\ncsmosek}; 
        \addplot [set19c3,dotted, name path=Bm] table [x={nx}, y={tmin}] {\ncsmosek}; 
        \addplot [set19c3,fill opacity=0.15] fill between [of=Am and Bm];
        \legend{Algorithm~\ref{alg:main} , SDP ($P$ diagonal), SDP ($P$ dense)};

    \end{axis}
\end{tikzpicture}